\documentclass[letterpaper, abstracton, DIV11, 12pt, titlepage=false]{scrartcl}

\usepackage[ruled]{algorithm2e}

\SetAlFnt{\small}
\SetAlCapFnt{\small}
\SetAlCapNameFnt{\small}
\SetAlCapHSkip{0pt}
\IncMargin{-\parindent}
\let\chapter\undefined

\usepackage{enumerate}
\usepackage[round]{natbib}
\usepackage[english]{babel}
\usepackage{amsfonts}
\usepackage{amsmath}
\usepackage[usenames,dvipsnames]{xcolor}
\usepackage{amssymb}
\usepackage{stmaryrd}
\usepackage{amsthm}
\usepackage{mathabx}
\usepackage{multirow}
\usepackage{multicol}
\usepackage{colortbl}
\usepackage{dsfont}
\usepackage{graphicx}
\usepackage{pifont}
\usepackage{bm}
\usepackage{pdfsync}
\usepackage{wrapfig}
\usepackage{microtype}
\usepackage{setspace}
\usepackage{pgf}
\usepackage{tikz}
\usetikzlibrary{arrows,automata}
\usepackage[colorlinks=true,bookmarks=true, pdftex, citecolor = blue, linkcolor = blue,urlcolor = blue]{hyperref}
\onehalfspace

\DisableLigatures[f]{encoding = *, family = * }

\def\bf{\normalfont\bfseries}
\changenotsign

\theoremstyle{plain}% default
\newtheorem{theorem}{Theorem}
\newtheorem*{theorem*}{Theorem}

\theoremstyle{definition}
\newtheorem{definition}{Definition}

\newcommand{\ourrep}{}

\theoremstyle{remark}

\begin{document}

{\setstretch{1}
\title{
\LARGE{%
Two New Impossibility Results \\ for the Random Assignment Problem}\thanks{
Department of Informatics, University of Zurich, 8050 Zurich, Switzerland, mennle@ifi.uzh.ch. 
We would like to thank (in alphabetical order)
Umut Dur, 
Lars Ehlers, 
Albin Erlanson, 
Aris Filos-Ratsikas, 
and Steffen Schuldenzucker 
for fruitful discussions on this topic. 
Any errors remain our own. 
This research was supported by the Swiss National Science Foundation under grant \#156836.
}}
\author{%
Timo Mennle \\ University of Zurich
\and Sven Seuken \\ University of Zurich }
\date{First version: July 6, 2016 \\
This version: \today}
%\date{February 13, 2014}
\maketitle

\begin{abstract} 
In this note, we prove two new impossibility results for random assignment mechanisms:
\citet{Bogomolnaia2001ANewSolution} showed that no assignment mechanism can satisfy \emph{strategyproofness}, \emph{ordinal efficiency}, and \emph{symmetry} at the same time, 
and \citet{MennleSeuken2017PSP_WP} gave a decomposition of strategyproofness into the axioms 
\emph{swap monotonicity}, \emph{upper invariance}, and \emph{lower invariance}. 
For our first impossibility result, we show that \emph{upper invariance}, \emph{lower invariance}, \emph{ordinal efficiency}, and \emph{symmetry} are incompatible. 
This refines the prior impossibility result because it relaxes swap monotonicity. 
For our second impossibility result, we show that no assignment mechanism satisfies 
\emph{swap monotonicity}, 
\emph{lower invariance},
\emph{ordinal efficiency}, 
\emph{anonymity}, 
\emph{neutrality}, 
and \emph{non-bossiness}. 
By contrasts, the \emph{Probabilistic Serial} (\emph{PS}) mechanism that \citet{Bogomolnaia2001ANewSolution} introduced, satisfies these axioms when lower invariance is replaced by upper invariance. 
It follows that there cannot exists a lower invariant counterpart to PS. 

\end{abstract}

%D78% Positive Analysis of Policy Formulation and Implementation%
%H75 State and Local Government: Health, Education, Welfare, Public Pensions 
%
}
%
%\renewcommand{\baselinestretch}{1.33}
%\section{Introduction}
%\label{SEC:INTRODUCTION}

\section{Model}
\label{SEC:MODEL}
Let $N$ be a set of $n$ \emph{agents} and let $M$ be a set of $n$ \emph{objects}. 
Each agent $i \in N$ has a strict preference order $P_i$ over objects, where $P_i:a\succ b$ means that agent $i$ prefers object $a$ to object $b$.
The set of all preference orders is denoted by $\mathcal{P}$.
A \emph{preference profile} $\bm P = (P_1 , \ldots, P_n) \in \mathcal{P}^N$ is a collection of preference orders of all agents, and we denote by $P_{-i} = (P_1,\ldots, P_{i-1},P_{i+1},\ldots,P_n) \in \mathcal{P}^{N\backslash\{i\}}$ the collection of preference orders of all agents, except $i$.

A \emph{(random) assignment} is represented by a bi-stochastic matrix $x = (x_{i,j})_{i \in N, j \in M}$ (i.e., $x_{i,j} \in [0,1]$, $\sum_{i'\in N} x_{i',j} = 1$, and $\sum_{j'\in M} x_{i,j'} = 1$ for all $i\in N,j \in M$).
The entry $x_{i,j}$ is the probability that agent $i$ gets object $j$. 
An assignment is \emph{deterministic} if all agents get exactly one full object (i.e., $x_{i,j} \in \{0,1\}$ for all $i \in N, j \in M$).
For any agent $i$, the $i$th row $x_i = (x_{i,j})_{j \in M}$ of the matrix $x$ is called the \emph{assignment vector of $i$} (or \emph{$i$'s assignment} for short).
The Birkhoff-von Neumann Theorem and its extensions \citep{Budishetal2013DesignRandomAllocMechsTheoryAndApp} ensure that, given any random assignment, we can always find a lottery over deterministic assignments that implements these marginal probabilities.
Finally, let $X$ and $\Delta(X)$ denote the spaces of all deterministic and random assignments, respectively.

A \emph{(random assignment) mechanism} is a mapping $ \varphi: \mathcal{P}^N \rightarrow \Delta(X)$ that chooses an assignment based on a profile of reported preference orders.
$\varphi_i(P_i,P_{-i})$ is the assignment vector that agent $i$ receives if it reports $P_i$ and the other agents report $P_{-i}$.
A mechanism is \emph{deterministic} if it only selects deterministic assignments (i.e., $\varphi: \mathcal{P}^N \rightarrow X$).

\section{Properties of Mechanisms}
\label{SEC:PROPERTIES}
In this section, we present the axioms that are needed for the formulation of our impossibility results. 
For this purpose, let $\varphi$ be a mechanism, let $\bm P \in \mathcal{P}$ be a preference profile, let $i\in N$ be some agent, let $x,y \in \Delta(X)$ be two assignments, and let $x_i$ and $y_i$ be the assignment vector if $i$ in $x$ and $y$, respectively. 
\subsection{Ordinal Efficiency}
First, we define first order-stochastic dominance and ordinal efficiency, which is a refinement of ex-post efficiency.\footnote{An assignment is \emph{ex-post efficient} if it can be decomposed into Pareto undominated, deterministic assignments.} 
\begin{definition}[First Order-stochastic Dominance] 
\emph{$x_i$ first order-stochastically dominates $y_i$ at $P_i$} if, for all objects $j \in M$, 
\begin{equation}
	\sum_{j' \in M ~:~ P_i:j' \succ j } x_{i,j'} \geq \sum_{j' \in M ~:~ P_i:j' \succ j } y_{i,j'}. 
\label{EQ:FOSD_INEQU}
\end{equation}
This dominance is \emph{strict} if inequality (\ref{EQ:FOSD_INEQU}) is strict for at least one object $j\in M$. 

\emph{$x$ ordinally dominates $y$ at $\bm P$} if $x_i$ first order-stochastically dominates $y_i$ at $P_i$ for all agents $i\in N$, and this dominance is \emph{strict} if the dominance of $x_i$ over $y_i$ at $P_i$ is strict for at least one agent $i\in N$. 
\end{definition}
\begin{definition}[Ordinal Efficiency]
$x$ is \emph{ordinally efficient at $\bm P$} if there exists no other assignment $y \in \Delta(X)$ that strictly ordinally dominates $x$ at $\bm P$. 
$\varphi$ is \emph{ordinally efficient} if $\varphi(\bm P)$ is ordinally efficiency at $\bm P$ for all preference profiles $\bm P \in \mathcal{P}^N$. 
\end{definition}
\subsection{Incentive Properties}
In this section, we define the well known requirement of strategyproofness and three axioms into which strategyproofness can be decomposed. 
These axioms and the decomposition result were originally given in \citep{MennleSeuken2017PSP_WP}.  
\begin{definition}[Strategyproofness]
$\varphi$ is \emph{strategyproof} if, 
for 
all agents $i\in N$, 
all preference profiles $(P_i,P_{-i}) \in \mathcal{P}^N$, 
and all preference orders $P_i' \in N_{P_i}$, 
we have that $\varphi_i(P_i,P_{-i})$ first order-stochastically dominates $\varphi_i(P_i',P_{-i})$ at $P_i$. 
\end{definition}
The three axioms we define next restrict the way in which a mechanism can react to certain changes in an agent's preference report. 
Specifically, the axioms specify the mechanism's behavior when an agent inverts the order of two consecutively ranked objects. 
\begin{definition}[Neighborhood]
For any two preference orders $P, P' \in \mathcal{P}$ we say that $P$ and $P'$ are \emph{adjacent} if they differ by just a swap of two consecutive object; 
formally, 
\begin{eqnarray*}
P & : & j_1 \succ \ldots \succ j_{k-1} \succ  j_k \succ j_{k+1} \succ j_{k+2} \succ \ldots \succ j_m, \\
P' & : & j_1 \succ \ldots \succ j_{k-1} \succ j_{k+1} \succ j_{k} \succ j_{k+2} \succ \ldots \succ j_m.
\end{eqnarray*}
The set of all preference orders adjacent to $P$ is the \emph{neighborhood of $P$}, denoted $N_P$.
\end{definition}
\begin{definition}[Upper and Lower Contour Sets]
For an object $j\in M$ and a preference order $P\in \mathcal{P}$, the \emph{upper contour set $U(j,P)$} and the \emph{lower contour set $L(j,P)$ of $j$ at $P$} are the sets of objects that an agent with preference order $P$ strictly prefers to $j$ or likes strictly less than $j$, respectively; 
formally, $U(j,P) = \{j' \in M~|~P: j' \succ j \}$ and $L(j,P) = \{j' \in M~|~P: j \succ j' \}$.
\end{definition}
\begin{definition}[Swap Monotonicity]
$\varphi$ is \emph{swap monotonic} if, 
for all agents $i\in N$, 
all preference profiles $(P_i,P_{-i}) \in \mathcal{P}^N$, 
and all preference orders $P_i' \in N_{P_i}$ from the neighborhood of $P_i$ with $P_i: j \succ j'$ but $P_i': j' \succ j$, 
one of the following holds: 
\begin{itemize}
\setlength{\itemsep}{0pt}
	\item either $\varphi_i(P_i,P_{-i}) = \varphi_i(P_i',P_{-i})$, 
	\item or $\varphi_{i,j}(P_i,P_{-i}) > \varphi_{i,j}(P_i',P_{-i})$ and $\varphi_{i,j'}(P_i,P_{-i}) < \varphi_{i,j'}(P_i',P_{-i})$.
\end{itemize}
\end{definition}
\begin{definition}[Upper Invariance]
$\varphi$ is \emph{upper invariant} if 
for all agents $i\in N$, 
all preference profiles $(P_i,P_{-i}) \in \mathcal{P}^N$, 
and all preference orders $P_i' \in N_{P_i}$ from the neighborhood of $P_i$ with $P_i: j \succ j'$ but $P_i': j' \succ j$, 
we have that $i$'s assignment for objects from the upper contour set of $j$ does not change; 
formally
\begin{equation}
\varphi_{i,j''}(P_i,P_{-i}) = \varphi_{i,j''}(P_i',P_{-i})\text{ for all }j''\in U(j,P_i).
\end{equation} 
\end{definition}
\begin{definition}[Lower Invariance]
$\varphi$ is \emph{lower invariant} if, 
for all agents $i\in N$, 
all preference profiles $(P_i,P_{-i}) \in \mathcal{P}^N$, 
and all preference orders $P_i' \in N_{P_i}$ from the neighborhood of $P_i$ with $P_i: j \succ j'$ but $P_i': j' \succ j$, 
we have that $i$'s assignment for objects from the lower contour set of $j'$ does not change; 
formally
\begin{equation}
\varphi_{i,j''}(P_i,P_{-i}) = \varphi_{i,j''}(P_i',P_{-i})\text{ for all }j''\in L(j',P_i).
\end{equation} 
\end{definition}
%
%\citet{MennleSeuken2015PSP_WP} have shown that a mechanism is strategyproof if and only if it is swap monotonic, upper invariant, and lower invariant. 
%
%
\subsection{Fairness Properties}
In this section, we define several common fairness properties. 
\begin{definition}[Symmetry] 
$x$ is \emph{symmetric at $\bm P$} if $x_i = x_{i'}$ for all $i,i'\in N$ with $P_i=P_{i'}$. 
$\varphi$ is \emph{symmetric} if $\varphi(\bm P)$ is symmetric at $\bm P$ for any preference profile $\bm P \in \mathcal{P}^N$. 
\end{definition}
Symmetry is also sometimes referred to as equal treatment of equals. 
\begin{definition}[Anonymity] 
$\varphi$ is \emph{anonymous} if, for all pairs $(\bm P, \bm P') \in \mathcal{P}^N \times \mathcal{P}^N$ of preference profile with 
\begin{itemize}
\setlength{\itemsep}{0pt}
	\item $P_i = P_{i'}'$ and $P_{i'} = P_i'$ for two agents $i,i'\in N$, 
	\item $P_{i''} = P_{i''}'$ for all other agents $i'' \in N\backslash\{i,i'\}$, 
\end{itemize}
we have that 
$\varphi_i(\bm P) = \varphi_{i'}(\bm P')$. 
\end{definition}
Observe that anonymity implies symmetry but the converse does not hold. 
%
%\begin{definition}[Equal Treatment of Equal Upper Ranks] 
%$x$ satisfies \emph{equal treatment of equal upper ranks at $\bm P$} if 
%for any two agents $i,i'\in N$ and any object $j\in M$ with 
%$U(j,P_i) = U(j,P_{i'})$
%and 
%$P_{i}|_{U(j,P_i)} = P_{i'}|_{U(j,P_{i'})}$ 
%\edit{[rankings coincide up to $j$]}, 
%we have that $x_{i,j'} = x_{i',j'}$ for all $j'\in U(j,P_i)$. 
%$\varphi$ satisfies \emph{equal treatment of equal upper ranks (UET)} if $\varphi(\bm P)$ satisfies equal treatment of equal upper ranks at $\bm P$ for any preference profile $\bm P \in \mathcal{P}^N$. 
%\end{definition}
%%
%\begin{definition}[Equal Treatment of Equal Lower Ranks] 
%$x$ satisfies \emph{equal treatment of equal lower ranks at $\bm P$} if 
%for any two agents $i,i'\in N$ and any object $j\in M$ with 
%$L(j,P_i) = L(j,P_{i'})$
%and 
%$P_{i}|_{L(j,P_i)} = P_{i'}|_{L(j,P_{i'})}$ 
%\edit{[rankings coincide after $j$]}, 
%we have that $x_{i,j'} = x_{i',j'}$ for all $j'\in L(j,P_i)$. 
%$\varphi$ satisfies \emph{equal treatment of equal lower ranks (LET)} if $\varphi(\bm P)$ satisfies equal treatment of equal lower ranks at $\bm P$ for any preference profile $\bm P \in \mathcal{P}^N$. 
%\end{definition}
%
%
\subsection{Other Properties}
We now define two additional properties: \emph{neutrality} and \emph{non-bossiness}.
For a preference order $P \in \mathcal{P}$ and two objects $j,j' \in M$, let $P^{j\leftrightarrow j'}$ be the preference order such that
\begin{eqnarray*}
	P & : & j_1 \succ \ldots \succ j_{k-1} \succ j \succ j_{k+1} \succ \ldots \succ j_{k'-1} \succ j' \succ j_{k'+1} \succ \ldots \succ j_m, \\
	P^{j\leftrightarrow j'} & : & j_1 \succ \ldots \succ j_{k-1} \succ j' \succ j_{k+1} \succ \ldots \succ j_{k'-1} \succ j \succ j_{k'+1} \succ \ldots \succ j_m.
\end{eqnarray*}
In words, $P$ and $P^{j\leftrightarrow j'}$ coincide except that the objects $j$ and $j'$ have traded positions. 
For a preference profile $\bm P \in \mathcal{P}^N$ and two objects $j,j' \in M$ let $\bm{P^{j\leftrightarrow j'}}$ be the preference profile with the preference orders $P_i^{j\leftrightarrow j'}$ constructed in this way. 
\begin{definition}[Neutrality] 
$\varphi$ is \emph{neutral} if, 
for all preference profiles $\bm P \in \mathcal{P}^N$ and all pairs of objects objects $(j,j')\in M^2$, 
we have that 
$\varphi_{i,j}(\bm P) = \varphi_{i,j'}(\bm P^{j\leftrightarrow j'})$ for all $i \in N$.
\end{definition}
In words, under a neutral mechanism, the assignment is independent of the specific names of the objects. 
\begin{definition}[Non-bossiness] 
$\varphi$ is \emph{non-bossy} if, 
for all agents $i\in N$, 
all preference profiles $(P_i,P_{-i}) \in \mathcal{P}^N$, 
and all preference orders $P_i' \in N_{P_i}$ 
such that $\varphi_i(P_i,P_{-i}) = \varphi_i(P_i',P_{-i})$, 
we have that $\varphi(P_i,P_{-i}) = \varphi(P_i',P_{-i})$.
\end{definition}
\section{Prior Results}
\label{SEC:PRIOR}
In this section, we restate two prior results: 
our own decomposition of strategyproofness from \citep{MennleSeuken2017PSP_WP} and the impossibility result from \citep{Bogomolnaia2001ANewSolution} pertaining to the incompatibility of strategyproofness, ordinal efficiency, and symmetry. 
\begin{theorem*}[Theorem 1 in \citep{MennleSeuken2017PSP_WP}] 
A mechanism is strategyproof if and only if it is swap monotonic, upper invariant, and lower invariant.
\end{theorem*}
\begin{theorem*}[Theorem 2 in \citep{Bogomolnaia2001ANewSolution}] 
For $n\geq 4$, there exists no mechanism that satisfies 
strategyproofness, 
ordinal efficiency, 
and symmetry.
\end{theorem*}
%
%\begin{theorem*}[Theorem 2 in \citep{Hashimoto2013TwoAxiomaticApproachestoPSMechanisms}]
%A mechanism is 
%ordinally efficient, 
%SD-envy-free, 
%and upper invariant\footnote{Upper invariance is called \emph{weak invariance} in \citep{Hashimoto2013TwoAxiomaticApproachestoPSMechanisms}.} 
%if and only if 
%it is the Probabilistic Serial mechanism.
%\end{theorem*}

%\section{Swap Monotonicity of Probabilistic Serial}
%
%\begin{theorem}
%\label{THM:SM_PS}
%The Probabilistic Serial mechanism is swap monotonic. 
%\end{theorem}
%%
%\begin{proof}
%
%\end{proof}

\section{Two New Impossibility Results}

In \citep{MennleSeuken2017PSP_WP}, we have shown that an interesting relaxed notion of strategyproofness arises when lower invariance is dropped from strategyproofness. 
The \emph{Probabilistic Serial mechanism (PS)} \citep{Bogomolnaia2001ANewSolution} satisfies this relaxed incentive requirement. 
Specifically, it is 
swap monotonic, 
upper invariant, 
ordinally efficient, 
and symmetric. 
%Moreover, many other non-strategyproof assignment mechanisms satisfy swap monotonicity and upper invariance. 
This raises the question whether other interesting relaxed notions of strategyproofness can be obtained by dropping either swap monotonicity or upper invariance instead. 

The following two impossibility results show that such relaxations do not admit the construction of new and appealing mechanisms. 
\subsection{An Impossibility Result without Swap Monotonicity}
\begin{theorem} 
\label{THM:RESULT_1}
For $n\geq 4$, there exists no mechanism that satisfies 
upper invariance, 
lower invariance, 
ordinal efficiency, 
and symmetry. 
\end{theorem}
\begin{proof}
Assume towards contradiction that $\varphi$ is an upper invariant, lower invariant, ordinally efficient, symmetric mechanism. 
Consider a setting with four agents $N=\{1,2,3,4\}$ and four objects $M = \{a,b,c,d\}$. 
To derive a contradiction, we derive the assignments that $\varphi$ must produce at various preference profiles. 
Figure \ref{FIG:THM1:PROOF_STRUCTURE} shows the order in which we consider these preference profiles. 
Arrows indicate that we use information about the assignment at one preference profile to infer information about the assignment at another and the lightning indicates the contradiction. 
\begin{figure}%
\begin{center}
\begin{tikzpicture}[
	->,
	>=stealth',
	shorten >=1pt,
	auto,
	node distance=2cm, 
	semithick]
  \tikzstyle{every state}=[fill=none,draw=black,text=black]

  \node[state] 				(I)          	    		   				{$I$};
  \node[state]        (II)  		[below left of=I] 		{$II$};
  \node[state]        (III) 		[below right of=II] 				{$III$};
  \node[state]        (IV)   		[below left of=II]      		{$IV$};
  \node[state]        (VII)  		[below right of=III]  {$VII$};
  \node[state]        (VI) 			[above right of=VII]  {$VI$};
  \node[state]        (V)  			[below left of=VII] 	{$V$};
  \node[state]        (VIII)  	[below right of=VII]	{$VIII$};

  \path (I) 		edge node {} (II)
								edge node {} (VI)
        (II) 		edge node {} (III)
								edge node {} (IV)
        (III) 	edge node {} (V)
								edge node {} (VII)
        (IV) 		edge node {} (V)
        (VI)  	edge node {} (VIII)
								edge node {} (VII)
        (VII) 	edge node {} (VIII)
        (V) 		edge[-] node {\large $\lightning$} (VIII);
\end{tikzpicture}
\end{center}
\caption{Order of derivation of assignments}%
\label{FIG:THM1:PROOF_STRUCTURE}%
\end{figure}
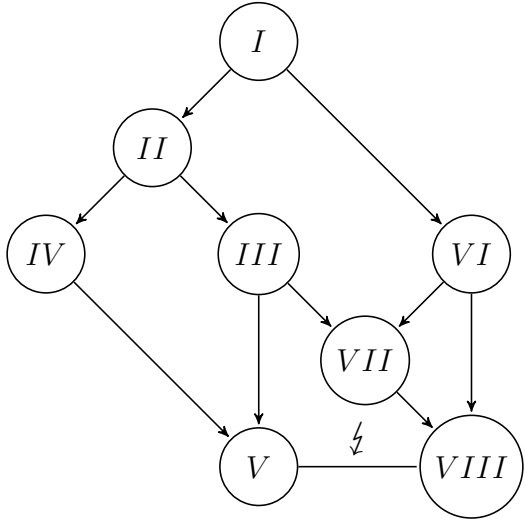

\begin{description}
	\item[$I$:] By symmetry, we have
		\begin{equation}
			\varphi(\bm P^{(I)}) = 
			\varphi\left(\begin{array}{c} 
				P_1^{(I)} ~:~ a \succ b \succ c \succ d \\ 
				P_2^{(I)} ~:~ a \succ b \succ c \succ d \\
				P_3^{(I)} ~:~ a \succ b \succ c \succ d \\
				P_4^{(I)} ~:~ a \succ b \succ c \succ d 
			\end{array}\right) 
			= 
			\left(\begin{array}{cccc}
				1/4 & 1/4 & 1/4 & 1/4 \\
				1/4 & 1/4 & 1/4 & 1/4 \\
				1/4 & 1/4 & 1/4 & 1/4 \\
				1/4 & 1/4 & 1/4 & 1/4  
			\end{array}\right).
		\end{equation}
	\item[$I \rightarrow II$:] If agent 4 swaps $c$ and $d$, we get 
		\begin{equation}
			\varphi(\bm P^{(II)}) = 
			\varphi\left(\begin{array}{c} 
				P_1^{(II)} ~:~ a \succ b \succ c \succ d \\ 
				P_2^{(II)} ~:~ a \succ b \succ c \succ d \\
				P_3^{(II)} ~:~ a \succ b \succ c \succ d \\
				P_4^{(II)} ~:~ a \succ b \succ d \succ c 
			\end{array}\right) 
			= 
			\left(\begin{array}{cccc}
				1/4 & 1/4 & 1/3 & 1/6 \\
				1/4 & 1/4 & 1/3 & 1/6 \\
				1/4 & 1/4 & 1/3 & 1/6 \\
				1/4 & 1/4 & 0 & 1/2  
			\end{array}\right),
		\end{equation}
		which follows from the following observations: 
		\begin{enumerate}
			\item 4's assignment for $a$ and $b$ cannot change by upper invariance. 
			\item By symmetry, the remaining probabilities for $a$ and $b$ must be evenly distributed across agents 1, 2, and 3. 
			\item If agent 4 had a strictly positive probability for $c$, agents 1, 2, and 3 would have zero probability for $d$, or else they could trade with 4 and improve their assignments in a first order-stochastic dominance sense. 
				However, this is not possible by ordinal efficiency. 
				Thus, $\varphi_{4,c}(\bm P^{(II)}) = 0$. 
			\item Since agent 4 has zero probability for $c$, it must receive $d$ with probability $1/2$ by feasibility. 
				The remaining probabilities for $c$ and $d$ must again be evenly distributed across agents 1, 2, and 3. 
		\end{enumerate}
	\item[$II \rightarrow III$:] Next, suppose that agent 4 swaps $b$ and $d$. 
		Then
		\begin{equation}
			\varphi(\bm P^{(III)}) = 
			\varphi\left(\begin{array}{c} 
				P_1^{(III)} ~:~ a \succ b \succ c \succ d \\ 
				P_2^{(III)} ~:~ a \succ b \succ c \succ d \\
				P_3^{(III)} ~:~ a \succ b \succ c \succ d \\
				P_4^{(III)} ~:~ a \succ d \succ b \succ c 
			\end{array}\right) 
			= 
			\left(\begin{array}{cccc}
				1/4 & 1/3 & 1/3 & 1/12 \\
				1/4 & 1/3 & 1/3 & 1/12 \\
				1/4 & 1/3 & 1/3 & 1/12 \\
				1/4 & 0 & 0 & 3/4  
			\end{array}\right)
		\end{equation}
		because
		\begin{enumerate}
			\item 4's assignment for $a$ and $c$ cannot change by upper and lower invariance. 
			\item If agent 4 had a strictly positive probability for $b$, agents 1, 2, and 3 would have positive probability for $d$, which contradicts ordinal efficiency. 
			\item This determines the assignment vector of agent 4. 
				By symmetry, the remaining probabilities have to be evenly distributed across agents 1, 2, and 3. 
		\end{enumerate}
	\item[$II \rightarrow IV$:] Going back to preference profile $\bm P^{(II)}$, suppose that agent 3 swaps $c$ and $d$. 
		Then
		\begin{equation}
			\varphi(\bm P^{(IV)}) = 
			\varphi\left(\begin{array}{c} 
				P_1^{(IV)} ~:~ a \succ b \succ c \succ d \\ 
				P_2^{(IV)} ~:~ a \succ b \succ c \succ d \\
				P_3^{(IV)} ~:~ a \succ b \succ d \succ c \\
				P_4^{(IV)} ~:~ a \succ b \succ d \succ c 
			\end{array}\right) 
			= 
			\left(\begin{array}{cccc}
				1/4 & 1/4 & 1/2 & 0 \\
				1/4 & 1/4 & 1/2 & 0 \\
				1/4 & 1/4 & 0 & 1/2 \\
				1/4 & 1/4 & 0 & 1/2  
			\end{array}\right)
		\end{equation}
		because
		\begin{enumerate}
			\item 3's assignment for $a$ and $b$ cannot change by upper invariance. 
			\item By symmetry, 4's assignment for $a$ and $b$ has to be $1/4$ each. 
			\item Completing the assignment matrix, 1's and 2's assignment for $a$ and $b$ must also be $1/4$ each. 
			\item If $3$ and 4 receive $c$ with positive probability, then $1$ and $2$ would receive $d$ with positive probability, which contradicts ordinal efficiency. 
		\end{enumerate}
	\item[$III \& IV \rightarrow V$:] 
		Now consider $\bm P^{(V)}$, which arises from $\bm P^{(III)}$ when agent 3 swaps $c$ and $d$. 
		Upper invariance implies that $\varphi_{3,a}(\bm P^{(V)}) = 1/4$ and $\varphi_{3,b}(\bm P^{(V)}) = 1/3$. 
		On the other hand, $\bm P^{(V)}$ also arises if agent 4 swaps $b$ and $d$ in $\bm P^{(IV)}$, and upper and lower invariance imply that $\varphi_{4,a}(\bm P^{(V)}) = 1/4$ and $\varphi_{4,c}(\bm P^{(V)}) = 0$. 
		Thus, 
		\begin{equation}
			\varphi(\bm P^{(V)}) = 
			\varphi\left(\begin{array}{c} 
				P_1^{(V)} ~:~ a \succ b \succ c \succ d \\ 
				P_2^{(V)} ~:~ a \succ b \succ c \succ d \\
				P_3^{(V)} ~:~ a \succ b \succ d \succ c \\
				P_4^{(V)} ~:~ a \succ d \succ b \succ c 
			\end{array}\right) 
			= 
			\left(\begin{array}{cccc}
				1/4 & 1/3 & 5/12 & 0 \\
				1/4 & 1/3 & 5/12 & 0 \\
				1/4 & 1/3 & 1/6 & 1/4 \\
				1/4 & 0 & 0 & 3/4  
			\end{array}\right). 
		\end{equation}
		To see the how the remaining entries are fixed, observe the following: 
		\begin{enumerate}
			\item Agent 4 cannot have positive probability for $b$. 
				Otherwise, some other agents would have positive probability for $d$, which contradicts ordinal efficiency. 
				This fixes $\varphi_{4,d}(\bm P^{(V)}) = 3/4$. 
			\item By symmetry, all agents must receive $a$ with probability $1/4$ and agents 1 and 2 must receive $b$ with probability $1/3$ each. 
			\item Agents 1 and 2 must have probability 0 for $d$. 
				Otherwise, ordinal efficiency would imply $\varphi_{3,c}(\bm P^{(V)}) = 0$ and thus $\varphi_{3,d}(\bm P^{(V)}) = 5/12$. 
				However, $\varphi_{3,d}(\bm P^{(V)}) + \varphi_{4,d}(\bm P^{(V)}) = 3/4 + 5/12 > 1$, a contradiction to feasibility. 
			\item The remaining entries follow from bi-stochasticity. 
		\end{enumerate}
	\item[$I \rightarrow VI$:] 
		Going back to $\bm P^{(I)}$, suppose that agent 3 swaps $a$ and $b$. 
		The resulting assignment is
		\begin{equation}
			\varphi(\bm P^{(VI)}) = 
			\varphi\left(\begin{array}{c} 
				P_1^{(VI)} ~:~ a \succ b \succ c \succ d \\ 
				P_2^{(VI)} ~:~ a \succ b \succ c \succ d \\
				P_3^{(VI)} ~:~ b \succ a \succ c \succ d \\
				P_4^{(VI)} ~:~ a \succ b \succ c \succ d 
			\end{array}\right) 
			= 
			\left(\begin{array}{cccc}
				1/3 & 1/6 & 1/4 & 1/4 \\
				1/3 & 1/6 & 1/4 & 1/4 \\
				0 & 1/2 & 1/4 & 1/4 \\
				1/3 & 1/6 & 1/4 & 1/4  
			\end{array}\right), 
		\end{equation}
		where the arguments are analogous to those for $I \rightarrow II$. 
	\item[$III \& VI \rightarrow VII$:] 
		Going back to $\bm P^{(III)}$, suppose that agent 3 swaps $a$ and $b$, then we get 
		\begin{equation}
			\varphi(\bm P^{(VII)}) = 
			\varphi\left(\begin{array}{c} 
				P_1^{(VII)} ~:~ a \succ b \succ c \succ d \\ 
				P_2^{(VII)} ~:~ a \succ b \succ c \succ d \\
				P_3^{(VII)} ~:~ b \succ a \succ c \succ d \\
				P_4^{(VII)} ~:~ a \succ d \succ b \succ c 
			\end{array}\right) 
			= 
			\left(\begin{array}{cccc}
				1/3 & 5/24 & 1/3 & 1/8 \\
				1/3 & 5/24 & 1/3 & 1/8 \\
				0 & 7/12 & 1/3 & 1/12 \\
				1/3 & 0 & 0 & 2/3  
			\end{array}\right). 
		\end{equation}
		To see why, observe the following: 
		\begin{enumerate}
			\item By lower invariance, 3's assignment for $c$ and $d$ cannot change relative to $\varphi_{3}(\bm P^{(III)})$. 
			\item Agent 3 may not get $a$ with positive probability; 
				otherwise, some other agent receives $b$ with positive probability, which contradicts ordinal efficiency. 
				This completely fixes the assignment of agent 3. 
			\item By ordinal efficiency and the fact that $\varphi_{3,d}(\bm P^{(VII)}) > 0$, agent 4 may not get $b$ or $c$ with positive probability. 
			\item If agent 4 swaps $d$ with $b$ and then with $c$, upper invariance requires that agent 4's assignment for $a$ does not change. 
				Since these changes induce preference profile $\bm P^{(VI)}$, we obtain $\varphi_{4,a}(\bm P^{(VII)}) = 1/3$. 
				This completely fixed the assignment of agent 4. 
			\item The assignments of 1 and 2 follow by symmetry. 
		\end{enumerate}
	\item[$VI \& VII \rightarrow VIII$:] 
		Starting with $\bm P^{(VII)}$, if agent 3 swaps $c$ and $d$, we get
		\begin{equation}
			\varphi(\bm P^{(VIII)}) = 
			\varphi\left(\begin{array}{c} 
				P_1^{(VIII)} ~:~ a \succ b \succ c \succ d \\ 
				P_2^{(VIII)} ~:~ a \succ b \succ c \succ d \\
				P_3^{(VIII)} ~:~ b \succ a \succ d \succ c \\
				P_4^{(VIII)} ~:~ a \succ d \succ b \succ c 
			\end{array}\right) 
			= 
			\left(\begin{array}{cccc}
				1/3 & 5/24  & 11/24 & 0 \\
				1/3 & 5/24  & 11/24 & 0 \\
				0   & 7/12 & x & 5/12-x \\
				1/3 & 0    & 1/12-x & 7/12+x
			\end{array}\right) 
		\end{equation}
		for some $x \in [0,1/12]$. 
		To see why, observe the following: 
		\begin{enumerate}
			\item By upper invariance, agent 3 retains the same probabilities for $a$ and $b$ as under $\bm P^{(VII)}$. 
			\item Starting at $\bm P^{(VI)}$, if agent 3 swaps $c$ and $d$, its assignment for $a$ cannot change by upper invariance, and by symmetry, agent 4 would still receive $a$ with probability $1/3$. 
				Again by upper invariance, agent 4 may then swap $d$ with $c$ and $b$, which cannot change agent 4's probability for $a$. 
				Thus, $\varphi_{4,a}(\bm P^{(VIII)}) = 1/3$. 
			\item If agent 4 had a positive probability for $b$, then some other agent would have positive probability for $d$, a contradiction to ordinal efficiency. 
				Thus, $\varphi_{4,b}(\bm P^{(VIII)}) = 0$. 
			\item The remaining probabilities for $a$ and $b$ follow by symmetry. 
			\item If agents 1 and 2 had positive probability for $d$, then neither 3 nor 4 could have positive probability for $c$ by ordinal efficiency. 
				Thus, agents 1 and 2 would have to absorb all probability for $c$ symmetrically. 
				However, each of them already receives either $a$ or $b$ with a total probability of $13/24$, which is greater than $1/2$, a contradiction. 
				Therefore, $\varphi_{1,d}(\bm P^{(VIII)}) = \varphi_{2,d}(\bm P^{(VIII)}) = 0$. 
			\item The assignment for $c$ of agents 1 and 2 is determined by their assignments for all other objects. 
			\item The remaining share of $c$ of $1/12$ must be absorbed by agents 3 and 4. 
				Specifically, let $\varphi_{3,c}(\bm P^{(VIII)}) = x$, where $x \in [0,1/12]$. 
		\end{enumerate}
\end{description}
Finally, observe that if agent 3 swaps $a$ and $b$ in $\bm P^{(VIII)}$, we arrive at $\bm P^{(V)}$. 
We get $\varphi_{3,c}(\bm P^{(V)}) = x \in [0,1/12]$ from upper invariance. 
However we previously observed that $\varphi_{3,c}(\bm P^{(V)}) = 1/6$, a contradiction. 

If there is a fifth agent ($5$, say) and a fifth object ($e$, say), we append $e$ at the ends of the preference orders of the agents $1, 2, 3, 4$ and let agent $5$ prefer object $e$ to all other objects. 
By ordinal efficiency, agent $5$ must receive object $e$ with certainty at all preference profiles, but the remaining relationships are unchanged. 
For any additional agents and objects, we proceed likewise. 
The contradiction can then be derived as in the case of four agents. 
\end{proof}
Theorem \ref{THM:RESULT_1} shows that no mechanism can be upper invariant, lower invariant, ordinally efficient, and symmetric. 
This refines the impossibility result of \citet{Bogomolnaia2001ANewSolution}, who showed that no mechanism can be strategyproof, ordinally efficient, and symmetric. 
Both upper and lower invariance are implied by strategyproofness but the opposite does not hold \citep{MennleSeuken2017PSP_WP}. 
For $n \leq 3$, ex-post efficiency coincides with ordinal efficiency. 
Since the Random Serial Dictatorship mechanism is 
strategyproof, 
ex-post efficient, 
and symmetric, 
it satisfies all requirements for $n\leq 3$. 
\subsection{An Impossibility Result without Upper Invariance}
In this section, we prove our second impossibility result. 
Recall that among swap monotonic and upper invariant mechanisms, PS is appealing as it satisfies 
ordinal efficiency, 
anonymity, 
neutrality, 
and non-bossiness. 
Our new result shows that it is impossible to design mechanisms \emph{like} PS which are swap monotonic and \emph{lower} invariant and have the same good axiomtic properties otherwise. 
\begin{theorem} 
\label{THM:RESULT_2}
For $n\geq 4$, there exists no mechanism that satisfies 
swap monotonicity, 
lower invariance, 
ordinal efficiency, 
anonymity, 
neutrality, 
and non-bossiness.
\end{theorem}
\begin{proof}
Assume towards contradiction that $\varphi$ is a 
swap monotonic, 
lower invariant, 
ordinally efficient, 
anonymous, 
neutral, 
and non-bossy mechanism.
Consider a setting with four agents $N=\{1,2,3,4\}$ and four objects $M = \{a,b,c,d\}$. 
To derive a contradiction, we derive the assignments that $\varphi$ must produce at various preference profiles. 
Figure \ref{FIG:THM2:PROOF_STRUCTURE} shows the order in which we consider these preference profiles. 
Arrows indicate that we use information about the assignment at one preference profile to infer information about the assignment at another and the lightning indicates the contradiction. 
\begin{figure}%
\begin{center}
\begin{tikzpicture}[
	->,
	>=stealth',
	shorten >=1pt,
	auto,
	node distance=2cm, 
	semithick]
  \tikzstyle{every state}=[fill=none,draw=black,text=black]

  \node[state] 				(I)          	    		   				{$I$};
  \node[state]        (II)  		[right of=I] 					{$II$};
  \node[state]        (III) 		[below of=I]	 				{$III$};
	\node[state]        (IV) 			[below of=II]	 				{$IV$};
	\node[state]        (V) 			[right of=IV]	 				{$V$};
	\node[state]        (VI) 			[right of=V]	 				{$VI$};
	\node[state]        (VII) 			[right of=VI]	 				{$VII$};

  \path (I) 		edge node {} (II)
        (II) 		edge node {} (IV)
        (III) 	edge node {} (IV)
        (IV) 		edge node {} (V)
        (V)  		edge node {} (VI)
        (VI) 		edge node {} (VII)
				(VII) 	edge [loop right] node {\large $\lightning$} ();
        %(V) 		edge[-] node {\large $\lightning$} (VIII);
\end{tikzpicture}
\end{center}
\caption{Order of derivation of assignments}%
\label{FIG:THM2:PROOF_STRUCTURE}%
\end{figure}
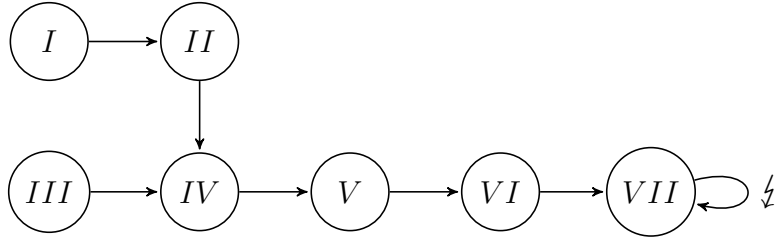
\begin{description}
	\item[$I$:] First, we show that
		\begin{equation}
			\varphi(\bm P^{(I)}) = 
			\varphi\left(\begin{array}{c} 
				P_1^{(I)} ~:~ a \succ b \succ c \succ d \\ 
				P_2^{(I)} ~:~ a \succ b \succ c \succ d \\
				P_3^{(I)} ~:~ b \succ a \succ d \succ c \\
				P_4^{(I)} ~:~ b \succ a \succ d \succ c 
			\end{array}\right) 
			= 
			\left(\begin{array}{cccc}
				1/2 & 0 & 1/2 & 0 \\
				1/2 & 0 & 1/2 & 0 \\
				0 & 1/2 & 0 & 1/2 \\
				0 & 1/2 & 0 & 1/2  
			\end{array}\right).
		\end{equation}
		Since anonymity implies symmetry, we find that the assignment must have the form
		\begin{equation}
			\varphi(\bm P^{(I)}) = 
			\varphi\left(\begin{array}{c} 
				P_1^{(I)} ~:~ a \succ b \succ c \succ d \\ 
				P_2^{(I)} ~:~ a \succ b \succ c \succ d \\
				P_3^{(I)} ~:~ b \succ a \succ d \succ c \\
				P_4^{(I)} ~:~ b \succ a \succ d \succ c 
			\end{array}\right) 
			= 
			\left(\begin{array}{cccc}
				x_{1,a} & x_{1,b} & x_{1,c} & x_{1,d} \\
				x_{1,a} & x_{1,b} & x_{1,c} & x_{1,d} \\
				x_{3,a} & x_{3,b} & x_{3,c} & x_{3,d} \\
				x_{3,a} & x_{3,b} & x_{3,c} & x_{3,d} 
			\end{array}\right).
		\end{equation}
		By anonymity, agents 1 and 2 may exchange their preference orders with agent 3 and 4. 
		The new assignment must have the form
		\begin{equation}
			\varphi(\bm P^{(I')}) = 
			\varphi\left(\begin{array}{c} 
				P_1^{(I')} ~:~ b \succ a \succ d \succ c \\
				P_2^{(I')} ~:~ b \succ a \succ d \succ c \\
				P_3^{(I')} ~:~ a \succ b \succ c \succ d \\ 
				P_4^{(I')} ~:~ a \succ b \succ c \succ d 
			\end{array}\right) 
			= 
			\left(\begin{array}{cccc}
				x_{3,a} & x_{3,b} & x_{3,c} & x_{3,d} \\
				x_{3,a} & x_{3,b} & x_{3,c} & x_{3,d} \\
				x_{1,a} & x_{1,b} & x_{1,c} & x_{1,d} \\
				x_{1,a} & x_{1,b} & x_{1,c} & x_{1,d} 
			\end{array}\right).
		\end{equation}
		By neutrality, we can rename objects as follows: $a \rightsquigarrow b$, $b \rightsquigarrow a$, $c \rightsquigarrow d$, and $d \rightsquigarrow c$. 
		The new assignment must have the form
		\begin{equation}
			\varphi(\bm P^{(I'')}) = 
			\varphi\left(\begin{array}{c} 
				P_1^{(I'')} ~:~ a \succ b \succ c \succ d \\
				P_2^{(I'')} ~:~ a \succ b \succ c \succ d \\
				P_3^{(I'')} ~:~ b \succ a \succ d \succ c \\ 
				P_4^{(I'')} ~:~ b \succ a \succ d \succ c 
			\end{array}\right) 
			= 
			\left(\begin{array}{cccc}
				x_{3,b} & x_{3,a} & x_{3,d} & x_{3,c} \\
				x_{3,b} & x_{3,a} & x_{3,d} & x_{3,c} \\
				x_{1,b} & x_{1,a} & x_{1,d} & x_{1,c} \\
				x_{1,b} & x_{1,a} & x_{1,d} & x_{1,c} 
			\end{array}\right).
		\end{equation}
		But since $\bm P^{(I)} = \bm P^{(I'')}$ we must have $x_{1,a} = x_{3,b}$, $x_{1,b} = x_{3,a}$, $x_{1,c} = x_{3,d}$, and $x_{1,d} = x_{3,c}$. 
		If $x_{3,c} > 0$, then $x_{1,d} > 0$, which violates ordinal efficiency. 
		Therefore, $x_{1,b} = x_{1,d} = x_{3,a} = x_{3,c} = 0$. 
		The remaining entries follow from symmetric distribution of the probabilities to the respective agents. 
	\item[$I \rightarrow II$:] Starting with $\bm P^{(I)}$ let agent 3 swap $a$ down in until it has reached the last position. 
		Since agent 3's probability for receiving $a$ is already 0, none of these swaps can further reduce this probability. 
		By swap monotonicity, none of the swaps can therefore change the agent 3's assignment at all. 
		By non-bossiness, the assignment remains unchanged for all agents. 
		Similarly, agent 4 may rank $a$ last without changing the assignment. 
		We obtain 
		\begin{equation}
			\varphi(\bm P^{(II)}) = 
			\varphi\left(\begin{array}{c} 
				P_1^{(II)} ~:~ a \succ b \succ c \succ d \\ 
				P_2^{(II)} ~:~ a \succ b \succ c \succ d \\
				P_3^{(II)} ~:~ b \succ d \succ c \succ a \\
				P_4^{(II)} ~:~ b \succ d \succ c \succ a  
			\end{array}\right) 
			= 
			\left(\begin{array}{cccc}
				1/2 & 0 & 1/2 & 0 \\
				1/2 & 0 & 1/2 & 0 \\
				0 & 1/2 & 0 & 1/2 \\
				0 & 1/2 & 0 & 1/2  
			\end{array}\right).
		\end{equation}
	\item[$III$:] The arguments to show that 
			\begin{equation}
			\varphi(\bm P^{(III)}) = 
			\varphi\left(\begin{array}{c} 
				P_1^{(III)} ~:~ a \succ b \succ c \succ d \\ 
				P_2^{(III)} ~:~ a \succ b \succ c \succ d \\
				P_3^{(III)} ~:~ d \succ b \succ c \succ a \\
				P_4^{(III)} ~:~ d \succ b \succ c \succ a 
			\end{array}\right) 
			= 
			\left(\begin{array}{cccc}
				1/2 & 1/4 & 1/4 & 0 \\
				1/2 & 1/4 & 1/4 & 0 \\
				0 & 1/4 & 1/4 & 1/2 \\
				0 & 1/4 & 1/4 & 1/2  
			\end{array}\right)
		\end{equation}
		are the same as those used to derive the assignment for $\bm P^{(I)}$. 
	\item[$II \& III \rightarrow IV$:] Next, we show that
			\begin{equation}
			\varphi(\bm P^{(IV)}) = 
			\varphi\left(\begin{array}{c} 
				P_1^{(IV)} ~:~ a \succ b \succ c \succ d \\ 
				P_2^{(IV)} ~:~ a \succ b \succ c \succ d \\
				P_3^{(IV)} ~:~ b \succ d \succ c \succ a \\
				P_4^{(IV)} ~:~ d \succ b \succ c \succ a 
			\end{array}\right) 
			= 
			\left(\begin{array}{cccc}
				1/2 & 1/8 & 3/8 & 0 \\
				1/2 & 1/8 & 3/8 & 0 \\
				0 & 3/4 & 1/4 & 0 \\
				0 & 0 & 0 & 1  
			\end{array}\right).
		\end{equation}
		This follows from the following observations:
		\begin{enumerate}
			\item $\bm P^{(IV)}$ arises from $\bm P^{(III)}$ when agent 3 swaps $d$ and $b$. 
				By lower invariance, agent 3's assignment for $a$ and $c$ may not change, so that 
				$\varphi_{3,a}(\bm P^{(IV)}) = 0$ and $\varphi_{3,c}(\bm P^{(IV)}) = 1/4$. 
			\item Similarly, $\bm P^{(IV)}$ arises from $\bm P^{(II)}$ when agent 4 swaps $b$ and $d$, and by lower invariance we get $\varphi_{4,a}(\bm P^{(IV)}) = 0$ and $\varphi_{4,c}(\bm P^{(IV)}) = 0$. 
			\item Agents 1 and 2 have no probability for receiving $d$. 
				Otherwise, agent 3 would trade its probability for $c$, a contradiction to ordinal efficiency. 
			\item If agent 4 had positive probability for $b$, ordinal efficiency would imply that agent 3 has no probability for $d$. 
				But then 4 would receive $d$ with probability 1, which contradicts the assumption that agent 4 has positive probability for $b$. 
				Thus, $\varphi_{4,b}(\bm P^{(IV)}) = 0$, which implies $\varphi_{4,d}(\bm P^{(IV)}) = 1$.
			\item Bi-stochasticity and symmetry imply the remaining probabilities. 
		\end{enumerate}
	\item[$IV \rightarrow V$:] Starting with $\bm P^{(IV)}$, let agent 3 rank $d$ in the last position. 
		Similar to the case $I \rightarrow II$, this does not change the assignment for anyone (by swap monotonicity and non-bossiness), so we get
		\begin{equation}
			\varphi(\bm P^{(V)}) = 
			\varphi\left(\begin{array}{c} 
				P_1^{(V)} ~:~ a \succ b \succ c \succ d \\ 
				P_2^{(V)} ~:~ a \succ b \succ c \succ d \\
				P_3^{(V)} ~:~ b \succ c \succ a \succ d \\
				P_4^{(V)} ~:~ d \succ b \succ c \succ a 
			\end{array}\right) 
			= 
			\left(\begin{array}{cccc}
				1/2 & 1/8 & 3/8 & 0 \\
				1/2 & 1/8 & 3/8 & 0 \\
				0 & 3/4 & 1/4 & 0 \\
				0 & 0 & 0 & 1  
			\end{array}\right).
		\end{equation}
	\item[$V \rightarrow VI$:] Starting with $\bm P^{(V)}$, let agent 3 swap $c$ and $a$. 
		We show that 
		\begin{equation}
			\varphi(\bm P^{(VI)}) = 
			\varphi\left(\begin{array}{c} 
				P_1^{(VI)} ~:~ a \succ b \succ c \succ d \\ 
				P_2^{(VI)} ~:~ a \succ b \succ c \succ d \\
				P_3^{(VI)} ~:~ b \succ a \succ c \succ d \\
				P_4^{(VI)} ~:~ d \succ b \succ c \succ a 
			\end{array}\right) 
			= 
			\left(\begin{array}{cccc}
				1/2 & 1/8 & 3/8 & 0 \\
				1/2 & 1/8 & 3/8 & 0 \\
				0 & 3/4 & 1/4 & 0 \\
				0 & 0 & 0 & 1  
			\end{array}\right).
		\end{equation}	
		This follows from the following observations:
		\begin{enumerate}
			\item By lower invariance, agent 3 gets $d$ with probability 0. 
			\item If agents 1 and 2 received $d$ with any positive probability, they could trade with agent 4, a contradiction to ordinal efficiency. 
				Thus, $\varphi_{1,d}(\bm P^{(VI)}) = \varphi_{2,d}(\bm P^{(VI)}) = 0$, $\varphi_{4,a}(\bm P^{(VI)}) = \varphi_{4,b}(\bm P^{(VI)}) = \varphi_{4,c}(\bm P^{(VI)}) = 0$, and $\varphi_{4,d}(\bm P^{(VI)}) = 1$. 
			\item If agent 3 received $a$ with positive probability, it could trade with agents 1 and 2 for probability for $b$, which contradicts ordinal efficiency. 
				Thus, $\varphi_{3,a}(\bm P^{(VI)}) = 0$. 
			\item Observe that the swap of $a$ and $c$ by agent 3 had no effect on 3's probability for obtaining $a$. 
				By swap monotonicity, 3's assignment can not change at all, which yields 
				$\varphi_{3,b}(\bm P^{(VI)}) = 3/4$ and $\varphi_{3,c}(\bm P^{(VI)}) = 1/4$. 
			\item The remaining probabilities are distributed symmetrically to agents 1 and 2.
		\end{enumerate}
	\item[$VI \rightarrow VII$:] Starting with $\bm P^{(VI)}$, let agent agent 3 swap $b$ and $a$, such that 
		\begin{equation}
			\bm P^{(VII)} = \left(\begin{array}{c} 
				P_1^{(VII)} ~:~ a \succ b \succ c \succ d \\ 
				P_2^{(VII)} ~:~ a \succ b \succ c \succ d \\
				P_3^{(VII)} ~:~ a \succ b \succ c \succ d \\
				P_4^{(VII)} ~:~ d \succ b \succ c \succ a 
			\end{array}\right). 
		\end{equation}
		By lower invariance, $\varphi_{3,c}(\bm P^{(VII)}) = 1/4$ and $\varphi_{3,d}(\bm P^{(VII)}) = 0$. 
		Symmetry implies $\varphi_{1,c}(\bm P^{(VII)}) = \varphi_{2,c}(\bm P^{(VII)}) = 1/4$ and $\varphi_{1,d}(\bm P^{(VII)}) = \varphi_{2,d}(\bm P^{(VII)}) = 0$ as well. 
		Thus $\varphi_{4,d}(\bm P^{(VII)}) = 1$ and $\varphi_{4,c}(\bm P^{(VII)}) = 1/4$, which is infeasible , a contradiction. 
\end{description}
The extension to more than 4 agents and objects is analogous to the same extension in Theorem \ref{THM:RESULT_1}. 
\end{proof}

% Bibstyle aea.bst version 2009.05.20

\end{document}